\documentclass[12pt,reqno]{amsart}
\usepackage{amsmath,amssymb,amsthm}
\usepackage[usenames]{color}
\usepackage[shortlabels]{enumitem}
\setlist[enumerate,1]{leftmargin=20pt}
\setlist[itemize,1]{leftmargin=15pt}
\usepackage{hyperref}
\usepackage{url}

\newenvironment{eatab}
 {\medskip\noindent\begin{minipage}{\textwidth}\normalfont\ttfamily
  \begin{tabbing}mmm\=mmm\=mmm\=mmm\=mmm\=mmm\=mmm\=mmm\=mmm%
               \=mmm\=mmm\=mmm\=mmm\kill}
 {\end{tabbing}\end{minipage}\medskip}

\newtheorem{theorem}{Theorem}[section]
\newtheorem{corollary}[theorem]{Corollary}
\newtheorem{lemma}[theorem]{Lemma}

\theoremstyle{definition}
\newtheorem{definition}[theorem]{Definition}

\newtheorem{remark}[theorem]{Remark}
\newtheorem{terminology}[theorem]{Terminology}

\DeclareMathOperator*{\PAR}{\mathbf{PAR}}

\newcommand\Arity{\ensuremath{\mathrm{Arity}}}
\newcommand\Bool{\ensuremath{\mathtt{Bool}}}
\newcommand\bx{\ensuremath{\bar x}}

\newcommand\Merge{\ensuremath{\mathrm{Merge}}}
\newcommand\Dec{\ensuremath{\mathtt{decrement}}}
\newcommand\Else{\ensuremath{\mathtt{else\ }}}

\newcommand\eps{\ensuremath{\varepsilon}}
\newcommand\false{\ensuremath{\bot}}
\newcommand\F{\ensuremath{\mathtt{Fire}}}
\newcommand\Fn{\ensuremath{\mathtt{Fire}^n}}
\newcommand\fn{\ensuremath{f^n}}
\newcommand\fnz{\ensuremath{f^n_0}}
\newcommand\fno{\ensuremath{f^n_1}}

\newcommand\If{\ensuremath{\mathtt{if\ }}}

\newcommand\Inc{\ensuremath{\mathtt{increment}}}
\newcommand\Instr{\ensuremath{\mathrm{Instr}}}
\newcommand\Instrn{\ensuremath{\mathrm{Instr}(n)}}
\newcommand\Inter{\ensuremath{\mathrm{Inter}}}
\newcommand\Intra{\ensuremath{\mathrm{Intra}}}
\renewcommand\k{\ensuremath{\kappa}}
\newcommand\lat{\ensuremath{\leftarrowtail}}
\newcommand\N{\ensuremath{\mathbb N}}
\newcommand\nil{\ensuremath{\mathtt{nil}}}
\newcommand\Num{\ensuremath{\mathtt{Num}}}
\newcommand\prl{\hspace{.5em}\parallel\hspace{.5em}}
\newcommand\PA{\ensuremath{\mathrm{Prog}(A)}}
\newcommand\PB{\ensuremath{\mathrm{Prog}(B)}}
\newcommand\PC{\ensuremath{\mathrm{Prog}(C)}}
\newcommand\R{\ensuremath{\mathbb R}}
\newcommand\Random{\ensuremath{\mathrm{Random}}}
\newcommand\Real{\ensuremath{\mathtt{Real}}}
\newcommand\qef{\hfill$\triangleleft$} 
\newcommand\s{\ensuremath{\sigma}}
\newcommand\sn{\ensuremath{\sigma_n}}
\newcommand\set[1]{\ensuremath{\{#1\}}}

\newcommand\Then{\ensuremath{\mathtt{then\ }}}
\newcommand\fnr{\ensuremath{f^n_{r_n}}}
\newcommand\tnz{\ensuremath{t^n_0}}
\newcommand\tno{\ensuremath{t^n_1}}
\newcommand\tnr{\ensuremath{t^n_{r_n}}}
\newcommand\true{\ensuremath{\top}}
\newcommand\U{\ensuremath{\Upsilon}}
\newcommand\Undon{\ensuremath{\mathrm{Undo}(n)}}
\newcommand\V{\ensuremath{\mathcal V}}
\newcommand\Voc{\ensuremath{\mathrm{Voc}}}
\newcommand\VA{\ensuremath{\mathrm{Voc}(A)}}
\newcommand\VB{\ensuremath{\mathrm{Voc}(B)}}

\title[Reversify]{Reversify any sequential algorithm}
\author[Yuri Gurevich]{Yuri Gurevich\\
{\scriptsize University of Michigan}}
\address{Computer Science and Engineering\\
University of Michigan\\
Ann Arbor, MI  48109, U.S.A}
\email{gurevich@umich.edu}
\thanks{Partially supported by the US Army Research Office under W911NF-20-1-0297}

\begin{document}

\begin{abstract}
To reversify an arbitrary sequential algorithm $A$, we gently instrument $A$ with bookkeeping machinery.
The result is a step-for-step reversible algorithm that mimics $A$ step-for-step and stops exactly when $A$ does.

Without loss of generality, we presume that algorithm $A$ is presented as an abstract state machine that is behaviorally identical to $A$.
The existence of such representation has been proven theoretically, and
the practicality of such representation has been amply demonstrated.
\end{abstract}
\maketitle
\thispagestyle{empty}

\begin{quote}\raggedleft\small\it
Darn the wheel of the world! Why must it\\ continually turn over? Where is the reverse gear?\\
--- Jack London
\end{quote}

\section{Introduction}
\label{sec:intro}

In 1973, Charles Bennett posited that an``irreversible computer can always be made reversible'' \cite[p.~525]{Bennett}.
To this end, he showed how to transform any one-tape Turing machine $M$ that computes a function $F(x)$, into a reversible three-tape Turing machine $M^R$ computing the function $x \mapsto (x,F(x))$.
First, $M^R$ emulates the computation of $M$ on $x$, saving enough information to ensure step-for-step reversibility. If and when the output is computed, the emulation phase ends,
and $M^R$ proceeds to erase all saved information with the exception of the input.

Bennett's construction shows that, in principle, every sequential algorithm, is reversifiable%
\footnote{We attempt to give a new useful meaning to the word reversify. To reversify an algorithm means to transform it into a reversible form (rather than to formulate it anew in verse, which is the current dictionary meaning of reversify).}.
In practice, you don't want to compile your algorithms to a one-tape Turing machine $M$ and then execute the three-tape Turing machine $M^R$.

It had been discussed in the programming community, in particular by Edsger Dijkstra  \cite[pp~351--354]{Dijkstra} and David Gries \cite[pp~265--274]{Gries}, which programs are reversible, but Bennett's reversification idea was either unknown to or neglected by programming experts.

The progress was led by physicists. They reversified Boolean circuits and other computation models. ``We have shown'', wrote Edward Fredkin and Tommaso Toffoli \cite[p~252]{FT}, ``that abstract systems having universal computing capabilities can be constructed from simple primitives which are invertible''.
The interest in reversible computations and especially in reversible circuit computations soared with the advent of quantum computing. This is related to the fact that pure (involving no measurements) quantum computations are reversible.
There are books on reversible computations \cite{Al-Rabadi, De Vos, Morita, Perumala}.
The International Conference on Reversible Computation will have its 13th meeting in 2021 \cite{RC}

In this paper, we use sequential abstract state machines, in short sequential ASMs, to address the problem of practical reversification of arbitrary sequential algorithms. Why ASMs? Let us explain.

ASMs were introduced to faithfully simulate arbitrary algorithms on their natural abstraction levels \cite{G092}. One instructive early result was the formalization of the C programming language \cite{G098}.

In \cite{G103}, an ambitious ASM thesis was formulated: For every algorithm $A$, there is an ASM $B$ that is behaviorally equivalent to $A$.
If $A$ is sequential, then $B$ has the same initial states as $A$ and the same state transition function.
In \cite{G141}, we axiomatized sequential algorithms and proved the ASM thesis for them%
\footnote{Later these results were generalized to other species of algorithms, e.g.\ to synchronous parallel algorithms \cite{G157} and interactive algorithms \cite{G182}.}.
Thus, semantically, sequential algorithms are sequential ASMs.
In the meantime, substantial evidence has been accumulated to support the practicality of faithful ASM modeling. Some of it is found in the 2003 book \cite{Boerger}.

The main result of the present paper is a simple construction, for every sequential ASM $A$, of a reversible sequential ASM $B$ that step-for-step simulates $A$ and stops when $A$ does.
$B$ does exactly what $A$ does plus some bookkeeping.
If $A$ uses some input and output variables and computes some function, then $B$ uses the same input and output variables and computes the same function.

\subsection*{Acknowledgment}

Many thanks to Andreas Blass for generous sanity check.

\section{Preliminaries}
\label{sec:prelim}

The purpose of this section is to make the current paper self-contained.

\subsection{Sequential algorithms}
\label{sub:seq}

By sequential algorithms we mean algorithms as the term was understood before modern computer science generalized the notion of algorithm in various directions,
which happened in the final decades of the 20th century.
In this connection, sequential algorithms are also called classical.

While the term ``sequential algorithm" is short and convenient, it also is too laconic. Some explication is in order.
``Algorithms,'' said Andrei Kolmogorov in a 1953 talk \cite{Kolmogorov}, ``compute in steps of bounded complexity.'' Let's look more closely at the two aspects mentioned by Kolmogorov.
One aspect is computing in steps, one step after another. Kolmogorov didn't say ``one step after another.'' He didn't have to. That was understood at the time.

The other aspect is a somewhat vague constraint: the bounded complexity of any one step of the algorithm. We prefer a related constraint, arguably a version of Kolmogorov's constraint: the bounded resources of any one step of the algorithm. The bounded resources constraint, still informal, seems to us clearer and more suitable. It might have been Kolmogorov's intention all along. We do not know exactly what Kolmogorov said during that talk%
\footnote{Vladimir Uspensky, who chaired the Logic Department of Moscow State University after Kolmogorov's death, admitted to me that the abstract \cite{Kolmogorov} of Kolmogorov's talk for the Moscow Mathematical Society was written by him (Uspensky) after many unsuccessful attempts to squeeze an abstract from Kolmogorov.}.

To summarize, sequential algorithms can be characterized informally as transition systems that compute in bounded-resources steps, one step after another.

In our axiomatization of sequential algorithms \cite{G141}, the bounded resources constraint gives rise to the crucial bounded-exploration axiom. It is also used to justify that a sequential algorithm doesn't hang forever within a step; time is a bounded resource.

\smallskip
In the following subsections, we recall some basic notions of mathematical logic in the form appropriate to our purposes.

\subsection{Vocabularies}\mbox{}
\label{sub:syn}

A \emph{vocabulary} is a finite collection of function symbols where each symbol $f$ is endowed with some metadata according to the following clauses (V1)--(V4). We interleave the four clauses with auxiliary definitions and explanations.

\smallskip
\begin{enumerate}[leftmargin=30pt]
\item[(V1)] Each symbol $f$ is assigned a natural number, the \emph{arity} of $f$. \qef
\end{enumerate}

\smallskip
Define \emph{terms} (or \emph{expressions}) by induction.
If $f$ is an $r$-ary symbol in and $t_1,\dots,t_r$ are terms, then $f(t_1,\dots,t_r)$ is a term.
(The case $r=0$ is the basis of induction.)

\smallskip
\begin{enumerate}[leftmargin=30pt]
\item[(V2)] Some symbols $f$ are marked as \emph{relational}.
\end{enumerate}

\smallskip
Clauses~(V1) and (V2) are standard in logic, except that, traditionally, relations are viewed as separate category, not as a special functions.

\smallskip
\begin{enumerate}[leftmargin=30pt]
\item[(V3)] $f$ may be marked as \emph{dynamic}; if not then $f$ is called \emph{static}.  Nullary static symbols are called \emph{constants}; nullary dynamic symbols are called \emph{variables}.
\end{enumerate}

\smallskip
Clause~(V3) is related to our use of structures as states of algorithms. The intention is that, during computation, only dynamic functions may be assigned new values. We say that a term is \emph{static} if it involves only static functions.

\smallskip
We presume that every vocabulary contains the following \emph{obligatory} symbols which are all static.
\begin{itemize}
\item Constants \true\ and  \false\ (read ``true'' and ``false''), unary \Bool, and the standard propositional connectives. All these symbols are relational.
\item Constant 0, and unary \Num, \Inc, and \Dec. Of these four symbols, only \Num\ is relational.

\item Constant \nil\ (called \texttt{undef} in \cite{G141} and other early papers) and the (binary) equality sign =. Of these two symbols, only the equality sign is relational.
\end{itemize}

\smallskip
\begin{enumerate}[leftmargin=30pt]
\item[(V4)] Every dynamic symbol $f$ is assigned a static term, the \emph{default term} of $f$. If $f$ is relational then so is its default term.
\end{enumerate}

\smallskip
Clauses (V2) and (V4) constitute rudimentary typing which is sufficient for our purposes in this paper.
As a rule, the default term for any relational symbol is \false.
If a variable $v$ is supposed to take numerical values, then typically the default term for $v$ would be 0, but it could be 1.
This concludes the definition of vocabularies.

\smallskip
If \U\ and $\U'$ are vocabularies, we write
 $\U \subseteq \U'$, and we say that $\U$ is \emph{included} in $\U'$ and that $\U'$ \emph{includes} or \emph{extends} \U, if every \U\ symbol belongs to $\U'$ and has the same metadata in $\U'$.

\subsection{Structures}\mbox{}
\label{sub:sem}

A {\em structure} $X$ of vocabulary $\U$ is a nonempty set $|X|$, the \emph{universe} or \emph{base set} of $X$, together with interpretations of the function symbols in
$\U$.
The vocabulary \U\ may be denoted $\Voc(X)$.

An $r$-ary function symbol $f$ is interpreted as a function $f: |X|^r \to |X|$ and is called a \emph{basic function} of $X$.
If $f$ is nullary then $f$ is just a name of an element of (the universe of) $X$.
If $f$ is dynamic and $d$ is the default term for $f$, then the value (denoted by) $d$ is the \emph{default value} of $f$.

If $f$ is relational, then the elements $\true$ are $\false$ are the only possible values of $f$.
If $f(\bx)=\true$, we say that $f$ is true (or holds) at \bx; otherwise we say that $f$ is false (or fails) at \bx.
If $f,g$ are relations of the same arity $r$, then $f,g$ are \emph{equivalent} in $X$ if their values at every $r$-tuple of elements of $X$ are the same.

Any basic relation $f$ is the characteristic function of the set \set{x: f(x) = \true}. It is often convenient to treat $f$ as that set. We will do that in \S\ref{sec:ex}.

\begin{remark}[Names and denotations]
Syntactic objects often denote semantical objects. For example, vocabulary symbols denote basic functions.
Different conventions may be used for disambiguation, e.g.\ a basic function may be denoted $f_X$. We will use no disambiguation convention in this paper. It should be clear from the context whether a symbol means a syntactic or semantic object. \qef
\end{remark}

The equality sign has its usual meaning.
\Bool\ comprises (the values of) \true, \false\  which, together with the propositional connectives, form a two-element Boolean algebra.

Given a structure $X$, the \emph{value} $\V_X \big(f(t_1,\dots,t_r)\big)$ of a $\Voc(X)$ term $f(t_1,\dots,t_r)$ in $X$ is defined by induction:
\begin{equation}\label{eval}
\V_X \big(f (t_1,\dots,t_r) \big) =
 f\big(\V_X (t_1), \dots, \V_X (t_r)\big).
\end{equation}
Again, the case $r=0$ is the base of induction.

Instead of $\Inc(x)$, we will write $x+1$ and $x-1$.
\Num\ comprises the values of terms $0, 0+1, (0+1)+1, \dots$ which are all distinct. These values are denoted $0, 1, 2, \dots$ respectively;
we call them the \emph{natural numbers} of structure $X$, and we say that these values are \emph{numerical}.
\Dec\ is interpreted as expected as well.
Instead of $\Dec(x)$, we write $x-1$.
$\Dec(0) = \nil$.
The value of \nil\ is neither Boolean nor numerical.

\begin{remark}[Totality]
In accordance with \S\ref{sub:seq}, all basic functions are total. In applications, various error values may arise, in particular \emph{timeout}. But, for our purposes in this paper (as in \cite{G141}), an error value is just another value.
\end{remark}

A {\em location} in a structure $X$ is a pair $\ell = (f,\bx)$ where $f$ is a dynamic symbol in $\Voc(X)$ of some arity $r$ and \bx\ is an $r$-tuple of elements of $X$.
The value $f(\bx)$ is the \emph{content} of location $\ell$.

An {\em update of location} $\ell = (f,\bx)$ is a pair $(\ell,y)$, also denoted $(\ell \lat y)$, where $y$ an element of $X$; if $f$ is relational then $y$ is Boolean.
To \emph{execute} an update $(\ell\ \lat y)$ in $X$, replace the current content $\V_X(f(\bx))$ of $\ell$ with $y$, i.e., set $f_X(\bx)$ to $y$.
An update $(\ell \lat y)$ of location $\ell = (f,\bx)$ is \emph{trivial} if $y = f(\bx)$.

An \emph{update of structure} $X$ is an update of any location in $X$.
A set $\Delta$ of updates of $X$ is \emph{contradictory} if it contains updates $(\ell \lat y_1)$ and $(\ell \lat y_2)$ with distinct  $y_1, y_2$; otherwise $\Delta$ is \emph{consistent}.

\subsection{Sequential abstract state machines}
\label{sub:asm}\mbox{}

\smallskip
Fix a vocabulary \U\ and restrict attention to function symbols in \U\ and terms over \U.

\begin{definition}[Syntax of rules]
\emph{Rules} over vocabulary $\U$ are defined by induction.
\begin{enumerate}
\item An \emph{assignment rule} or simply \emph{assignment} has the form\quad
    \begin{equation}\label{assign}
    f(t_1,\dots,t_r):=t_0
    \end{equation}
    where $f$, the \emph{head} of the assignment, is dynamic, $r = \Arity(f)$, and $t_0,\dots,t_r$ are terms. If $f$ is relational, then the head function of $t_0$ is relational. The assignment~\eqref{assign} may be called an $f$ assignment.
\item A \emph{conditional rule} has the form
    \begin{equation}\label{ite}
    \If\ \beta\ \Then\ R_1\ \Else\ R_2
    \end{equation}
    where $\beta$ is a Boolean-valued term and $R_1, R_2$ are \U\ rules.
\item A \emph{parallel rule} has the form
    \begin{equation}\label{par}
    R_1 \prl R_2 \prl\cdots\prl R_k
    \end{equation}
    where $k$ is a natural number and $R_1, \dots, R_k$ are \U\ rules. In case $k=0$, we write \texttt{Skip}. \qef
\end{enumerate}
\end{definition}

\begin{definition}[Semantics of rules]
Fix an \U\ structure $X$.
Every $\U$ rule $R$ generates a finite set $\Delta$ of updates in $X$.
$R$ \emph{fails} in $X$ if $\Delta$ is contradictory;
otherwise $R$ \emph{succeeds} in $X$.
To \emph{fire} (or \emph{execute}) rule $R$ that succeeds in structure $X$ means to execute all $\Delta$ updates in $X$.

\begin{enumerate}
\item An assignment $f(t_1,\dots,t_r):=t_0$ generates a single update $(\ell,\V_X(t_0))$ where
    $ \ell = \big(f, (\V_X(t_1), \dots, \V_X(t_r)\big) $.
\item A conditional rule\quad $\If \beta\ \Then R_1\ \Else R_2$\quad works exactly as $R_1$, if $\beta =\true$ in $X$, and exactly as $R_2$ otherwise.
\item A parallel rule $R_1 \prl R_2 \prl\cdots\prl R_k$ generates the union of the update sets generated by rules $R_1, \dots, R_k$ in $X$. \qef
\end{enumerate}
\end{definition}

\begin{definition}
A \emph{sequential ASM} $A$ is given by the following three components.
\begin{enumerate}
\item A vocabulary \U, denoted $\Voc(A)$.
\item A nonempty collection of $\Voc(A)$ structures, closed under isomorphisms. These are the \emph{initial states} of $A$. \qef
\item A $\Voc(A)$ rule, called the \emph{program} of $A$ and denoted \PA.
 \end{enumerate}
\end{definition}
\noindent

As we mentioned in \S\ref{sec:intro}, every sequential algorithm $A$ is behaviorally identical to some sequential ASM $B$; they have the same initial states and the same state-transition function.

In the rest of the paper, by default, all ASMs are sequential.

Consider an ASM $A$.
A $\Voc(A)$ structure $X$ is \emph{terminal} for $A$ if \PA\ produces no updates (not even trivial updates%
\footnote{In applications, trivial updates of $A$ may mean something for its environment.}) in $X$.
A \emph{partial computation} of $A$ is a finite sequence $X_0, X_1, \dots, X_n$ of $\U$ structures where
\begin{itemize}
\item $X_0$ is an initial state of $A$,
\item every $X_{i+1}$ is obtained by executing \PA\ in $X_i$, and
\item no structure in the sequence, with a possible exception of $X_n$, is terminal.
\end{itemize}
If $X_n$ is terminal, then the partial computation is \emph{terminating}.
A \emph{(reachable) state} of $A$ is an $\Voc(A)$ structure that occurs in some partial computation of $A$.

A Boolean expression $\gamma$ is a \emph{green light} for an ASM $A$ if it holds in the nonterminal states of $A$ and fails in the terminal states.

\begin{lemma}\label{lem:green}
Every ASM has a green light.
\end{lemma}

\begin{proof}
By induction on rule $R$, we construct a green light $\gamma_R$ for any ASM with program $R$. If $R$ is an assignment, set $\gamma_R =\true$. If $R$ is the parallel composition of rules $R_i$, set $\gamma_R = \bigvee_i \gamma_{R_i}$. If $R = \If\ \beta\ \Then\ R_1\ \Else\ R_2$, set
$\gamma_R = \big(\beta \land \gamma_{R_1}\big) \lor
              \big(\neg\beta \land \gamma_{R_2}\big)$.
\end{proof}

In examples and applications, typically, such conditions are easily available. Think of terminal states of finite automata or of halting control states of Turing machines. In a while-loop program, the while condition is a green light.

\section{Reducts and expansions}
\label{sec:expand}

In mathematical logic, a structure $X$ is a \emph{reduct} of a structure $Y$ if $\Voc(X) \subseteq \Voc(Y)$, the two structures have the same universe, and every $\Voc(X)$ symbol $f$ has the same interpretations in $X$ and in $Y$.
If $X$ is a reduct of $Y$, then $Y$ is an \emph{expansion} of $X$.
For example, the field of real numbers expands the additive group of real numbers.

We say that an expansion $Y$ of a structure $X$ is \emph{uninformative} if the additional basic functions of $Y$ (which are not basic functions of $X$) are dynamic and take only their default values in $Y$. (The default values are defined in \S\ref{sub:sem}.) Clearly, $X$ has a unique uninformative expansion to \VB.

\begin{definition}\label{def:expand}
An ASM $B$ is a \emph{faithful expansion} of an ASM $A$ if the following conditions hold.
\begin{enumerate}[leftmargin=30pt]
\item[(E1)] $\VA\subseteq\VB$. The symbols in \VA\ are the \emph{principal} symbols of \VB, and their interpretations in \VB\ structures are \emph{principal} basic functions;
    the other \VB\ symbols and their interpretations are \emph{ancillary}.
\item[(E2)] All ancillary symbols are dynamic, and the initial states of $B$ are the uninformative expansions of the initial states of $A$.
\item[(E3)] If $Y$ is a $\Voc(B)$ structure and $X$ the $\Voc(A)$ reduct of $Y$, then the principal-function updates (including trivial updates) generated by \PB\ in $Y$ coincide with the those generated by \PA\ in $X$, and the ancillary-function updates generated by \PB\ in $Y$ are consistent. \qef
\end{enumerate}
\end{definition}

\begin{corollary}
Suppose that $B$ is a faithful expansion of an ASM $A$, then the following claims hold.
\begin{enumerate}
\item If $X_0, \dots, X_n$ is a partial computation of $A$ then there is a unique partial computation $Y_0, \dots, Y_n$ of $B$ such that every $X_i$ is the $\Voc(A)$ reduct of the corresponding $Y_i$.
\item If states $Y_0, \dots, Y_n$ of $B$ form a partial computation of $B$, then their \VA\ reducts form a partial computation $X_0, \dots, X_n$ of $A$.
\item The $\Voc(A)$ reduct of a state of $B$ is a state of $A$.
\end{enumerate}
\end{corollary}

If an ASM $A$ computes a function $F$, one would expect that any faithful expansion of $A$ computes function $F$ as well. To confirm this expectation, we need to formalize what it means to compute a function. In the context of this paper, every ASM state is endowed with a special copy of the set \N\ of natural numbers. This makes the desired formalization particularly easy for numerical partial functions $F: \N^k\to \N$.

\begin{corollary}\label{cor:fun}
Suppose that an ASM $A$ computes a partial numerical function $F: \N^k\to \N$ in the following sense:
\begin{enumerate}
\item $A$ has input variables $\iota_1, \dots, \iota_n$ taking numerical values in the initial states, and $A$ has an output variable $o$,
\item all initial states of $A$ are isomorphic except for the values of the input variables, and
\item the computation of $A$ with initial state $X$ eventually terminates if and only if $F$ is defined at tuple $\bx = \big(\V_X(\iota_1), \dots, \V_X(\iota_n)\big)$, in which case the final value of $o$ is $F(\bx)$.
\end{enumerate}
Then every faithful expansion of $A$ computes $F$ in the same sense. \qef
\end{corollary}

Corollary~\ref{cor:fun} can be generalized to computing more general functions and to performing other tasks, but this is beyond the scope of this paper.

An ASM may be faithfully expanded by instrumenting its program for monitoring purposes. For example, if you are interested how often a particular assignment \s\ fires, replace \s\ with a parallel composition
$$\s \prl \k:=\k+1$$
where a fresh variable $\k$, initially zero, is used as a counter. A similar counter is used in our Reversibility Theorem below.

\section{Reversibility}
\label{sec:rev}

\begin{definition}
An ASM $B$ is \emph{reversible (as is)} if there is an ASM $C$ which reverses all $B$'s computations in the following sense.
If $Y_0, Y_1, \dots, Y_n$ is a partial computation of $B$, then $Y_n, Y_{n-1}, \dots, Y_0$ is a terminating computation of $C$.
\end{definition}

\begin{theorem}[Reversification Theorem]\label{thm:rev}
Every ASM $A$ has a faithful reversible expansion.
\end{theorem}

\begin{proof}
Enumerate the (occurrences of the) assignments in \PA\ in the order they occur:
\[ \s_1, \s_2, \dots, \s_N\]
It is possible that $\s_i, \s_j$ are identical even though $i\ne j$. The metavariable $n$ will range over numbers $1,2,\dots,N$.
For each $n$, let \fn\ be the head of \sn, $r_n = \Arity(\fn)$, and \tnz, \tno, \dots, \tnr\ the terms such that
\[ 
\sn = \Big(\fn(\tno,\dots,\tnr) := \tnz\Big).
\] 

We construct an expansion $B$ of $A$. The ancillary symbols of $B$ are as follows.
\begin{enumerate}
\item A variable \k.
\item For every $n$, a unary relation symbol $\Fn$.
\item For every $n$, unary function symbols
\fnz, \fno, \dots, \fnr.
\end{enumerate}

The default term for \k\ is 0. The default term for all relations $\Fn$ is \false.
The default term for all functions \fnz, \fno, \dots, \fnr\ is \nil.
Accordingly, the initial states of $B$ are obtained from the initial states of $A$ by setting $\k=0$, every $\Fn(x) = \false$, and every $f^n_i(x) = \nil$,

The intention is this. If $X_0, X_1, \dots, X_l$ is a partial computation of $B$, then for each $k=0,\dots,l$ we have the following.
\begin{enumerate}
\item The value of \k\ in $X_k$ is $k$, so that \k\ counts the number of steps performed until now; we call it a \emph{step counter}.
\item $\Fn(\k)$ holds in $X_{k+1}$ if and only if \sn\ fires in $X_k$.
\item The values of $\fno(\k), \dots, \fnr(\k)$ in $X_{k+1}$ record the values of the terms \tno, \dots, \tnr\ in $X_k$ respectively, and the value of $\fnz(\k)$ in $X_{k+1}$ records the value of the term $\fn(\tno,\dots,\tnr)$ in $X_k$.
\end{enumerate}

The program of $B$ is obtained from \PA\ by replacing every assignment \sn\ with $\Instrn$ (an allusion to ``instrumentation'') where

\begin{eatab}
\> \Instrn\ = \\[2pt]
\>\>\> $\sn\ \prl \k:=\k+1 \prl \Fn(\k+1):=\true \prl$\\[2pt]
\>\>\> $\fnz(\k+1) := f^n(\tno,\dots,\tnr) \prl$ \\[2pt]
\>\>\> $\fno(\k+1) := \tno \prl \dots \prl \fnr(\k+1) := \tnr$
\end{eatab}

It is easy to check that the conditions (E1)--(E3) of Definition~\ref{def:expand} hold, and $B$ is indeed a faithful expansion of $A$.
In particular, if $Y$ and $X$ are as in (E3) and $X$ is terminal, then no assignment \sn\ fires in $X$, and therefore no \Instrn\ fires in $Y$, so that $Y$ is terminal as well.

\begin{lemma}\label{lem:default}
If $Y_0, \dots, Y_k$ is a partial computation of $B$, then
\begin{itemize}
\item $\k=k$ in $Y_k$ and
\item if $j>k$ then $\Fn(j),\fnz(j), \dots, \fnr(j)$ have their default values in $Y_k$.
\end{itemize}
\end{lemma}

\begin{proof}[Proof of lemma]
Induction on $k$.
\end{proof}

Now, we will construct an ASM $C$ which reverses $B$'s computations.
The vocabulary of $C$ is that of $B$, and any $\Voc(C)$ structure is an initial state of $C$. The program of $C$ is

\[
\If\ \k>0\ \Then\ \Big(\k:=\k - 1 \prl
  \PAR_n \Undon\Big)
\]

\noindent
where $\PAR$ is parallel composition, $n$ ranges over $\{1,2, \dots, N\}$, and

\begin{eatab}
\ $\Undon =$\\[2pt]
\>\> if $\F^n(\k)=\true$ then \\[2pt]
\>\>\> $\F^n(\k):=\false \prl$\\[2pt]
\>\>\> $\fn\Big(\fno(\k),\dots,\fnr(\k)\Big) := f_0^n(\k)
    \prl \fnz(\k) := \nil \prl$\\[2pt]
\>\>\> $\fno(\k) := \nil \prl \dots\ \prl \fnr(\k) := \nil$
\end{eatab}

\begin{lemma}\label{lem:main}
Let $Y$ be an arbitrary nonterminal $\Voc(B)$ structure such that all functions $\Fn$ and $f^n_i$ have their default values at argument $k = \V_Y(\k)$ in $Y$.
If \PB\ transforms $Y$ to $Y'$, then \PC\ transforms $Y'$ back to $Y$, i.e., \PC\ undoes the updates generated by \PB\ and does nothing else.
\end{lemma}

\begin{proof}[Proof of lemma]
The updates generated by \PB\ in $Y$ are the updates generated by the rules \Instrn\ such that \sn\ fires in $Y$.
Since $k=\V_Y(\k)$, we have $\V_{Y'}(\k) = k+1>0$, and therefore \PC\ decrements \k. It also undoes the other updates generated by the rules $\Instrn$. Indeed, suppose that \sn\ fires in $Y$.

To undo the update $\Fn(k+1) \lat \true$, \PC\ sets $\Fn(k+1)$ back to $\false$.

To undo the update $\fn\big(\tno, \dots, \tnr\big) \lat \tnz$, generated by \sn\ itself,
\PC\ sets $\fn\Big(\fno(k+1),\dots,\fnr(k+1)\Big)$ to $f_0^n(k+1)$. Recall that $\fno(k+1),\dots,\fnr(k+1)$ record $\tno, \dots, \tnr$ in $Y$ and $\fnz(k+1)$ records the value of $f^n\big(\tno, \dots, \tnr\big)$ in $Y$.
Thus, \PC\ sets $f^n\big(\tno, \dots, \tnr\big)$ back to its value in $Y$.

To undo the updates of
$ \fnz(k+1), \fno(k+1),\ \dots,\ \fnr(k+1) $,
\PC\ sets $\fnz(k+1), \fno(k+1), \dots, \fnr(k+1)$ back to \nil.

Thus, being executed in $Y'$, \PC\ undoes all updates generated by the rules \Instrn\ in $Y$. A simple inspection of \PC\ shows that it does nothing else. Thus, \PC\ transforms $Y'$ to $Y$.
\end{proof}

Now suppose that $Y_0,\dots,Y_n$ is a computation of $B$, $k<n$, $Y=Y_k$, and $Y'=Y_{k+1}$. Then $Y$ is nonterminal and, by Lemma~\ref{lem:default}, all $\Fn(\k)$ and $f^n_i(\k)$ have their default values in $Y$. By Lemma~\ref{lem:main}, \PC\ transforms $Y_{k+1}$ to $Y_k$. The $Y_0$ is a terminal state of $C$.
Thus, $C$ reverses all $B$'s computations.
\end{proof}

The proof of Reversification Theorem uses notation and the form of \PB\ which is convenient for the proof. In examples and applications, notation and \PB\ can be simplified.

\begin{remark}[Notation]\label{rem:notat}
Let $\sn$ be an assignment $\big(g(\tno,\dots,t^n_r):= \tnz\big)$ so that $f^n$ is $g$. If \sn\ is the only $g$ assignment in \PA\ or if every other $g$ assignment $\sigma_m$ in \PA\ is just another occurrence of \sn,
then the ancillary functions $f^n_i$ may be denoted $g_i$; no confusion arises. \qef
\end{remark}

Recall that a green light for an ASM $A$ is a Boolean-valued expression that holds in the nonterminal states and fails in the terminal states.

\begin{remark}[Green light and step counter]\label{rem:elegant}
In \PB, every \Instrn\ has an occurrence of the assignment $\k := \k+1$. A green light for $A$ provides an efficient way to deal with this excess. Notice that $B$ increments the step counter exactly when the green light is on.

\smallskip\noindent
Case~1: $\PA$ has the form\quad $\If \gamma\ \Then\ (k:=k+1 \parallel \Pi)$.

In this case, $\gamma$ is a green light for $A$, and $A$ has already a step counter, namely $k$.
Without loss of generality, $k$ is the step counter \k\ used by \PB; if not, rename one of the two variables.
Notice that the assignment $\sigma_1 = (k:=k+1)$ needs no instrumentation.
There is no need to signal firings of $\sigma_1$ because $\sigma_1$ fires at every step.
And, when a step is completed, we know the previous value of the step counter; there is no need to record it.

Let $\Instr^-(\Pi)$ be the rule obtained from $\Pi$ by first replacing every assignment \sn\ with the rule \Instrn\ defined in the proof of the program, and then removing all occurrences of $k := k+1$.
Then the program
\[ \If \gamma\ \Then \big(k:=k+1 \parallel \Instr^-(\Pi)\big) \]
has only one occurrence of $k:=k+1$ and is equivalent to \PB.

\smallskip\noindent
Case~2: $\PA = \big(\If \gamma\ \Then \Pi\big)$ where $\gamma$ is a green light for $A$ and the step counter \k\ of \PB\ does not occur in \PA.

The modified program\quad
$ \If \gamma\ \Then \big(\k:=\k+1 \parallel \Pi\big) $,\\
where \k\ is the step counter of $B$, is a faithful expansion of \PA, and thus Case~2 reduces to Case~1.

\smallskip\noindent
Case~3 is the general case.

By Lemma~\ref{lem:green}, every ASM program has a green light.
If $\gamma$ is a green light for $A$ and $\Pi = \PA$, then the program\quad $\If \gamma\ \Then \Pi$\quad is equivalent to \PA, and thus Case~3 reduces to Case~2. \qef
\end{remark}

The rules \Instrn\ and \Undon, described in the proof of the theorem, are the simplest in the case when $r_n=0$. In such a case, \sn\ has the form $v:=t$ where $v$ is a variable,
so that $f^n=v$ and $\tnz = t$. Then

\begin{eatab}
\ $\Instrn =$\\[2pt]
\>\ $\sn \prl \k:=\k+1 \prl
    \Fn(\k+1):=\true \prl v_0(\k+1):=v$,\\[2pt]
\ $\Undon =$\\[2pt]
\>\>\> if $\F^n(\k)=\true$ then \\[2pt]
\>\>\>\> $\F^n:=\false \prl v:= v_0(\k) \prl v_0(\k):= \nil$. \end{eatab}

\begin{lemma}\label{lem:simple}
Suppose that an assignment \sn\ to a variable $v$ can fire only at the last step of $A$ and that the update generated by \sn\ is never trivial.
Then, \Instrn\ and \Undon\ can be simplified to
\begin{align*}
\Instrn &= &&\sn \prl \k:=\k+1\\
\Undon  &=  &&\If\ v\ne d\ \Then\ v:= d
\end{align*}
where $d$ is the default term for $v$ in \VA.
\end{lemma}

We do not assume that \sn\ fires at the last step of every computation of $A$, and so the expression $v\ne d$ is not necessarily a green light for $A$.

\begin{proof}
Suppose that \sn\ fires in state $Y$ of $B$. Then $Y$ is nonterminal, $v=d$ in $Y$, the next state $Y'$ is terminal, and $v\ne d$ in $Y'$. It is easy to see the simplified version of \Undon\ indeed undoes the updates generated by the simplified version of \Instrn\ in $Y$.
\end{proof}

\section{Examples}
\label{sec:ex}

To illustrate the reversification procedure of \S\ref{sec:rev}, we consider three simple examples. By the reversification procedure we mean not only the constructions in the proof of Reversification Theorem, but also Remarks~\ref{rem:notat} and \ref{rem:elegant} and Lemma~\ref{lem:simple}.
Unsurprisingly, in each case, the faithful reversible expansion produced by the general-purpose procedure can be simplified.

\subsection{Bisection algorithm}
\label{sub:bisect}\mbox{}

The well-known bisection algorithm solves the following  problem where \R\ is the field of real numbers.
Given a continuous function $F: \R\to\R$ and reals $a, b, \eps$ such that $F(a) < 0 < F(b)$ and $\eps>0$, find a real $c$ such that $|F(c)| < \eps$.
Here is a draft program for the algorithm:

\begin{eatab}
\> if\ \ $|F\big((a+b)/2\big)|\ge\eps$\>\>\>\>\>\>\> then\\[2pt]
\> \> if\>\>\ $F\big((a+b)/2\big) <0$\>\>\>\>\>\
   then $a := (a+b)/2$\\[2pt]
\> \> elseif\>\>\ $F\big((a+b)/2\big) >0$\>\>\>\>\>\
   then $b := (a+b)/2$
   \\[2pt]
\> elseif\>\>\ $c=\nil$ then $c := (a+b)/2$
\end{eatab}

The condition $c=\nil$ in the last line ensures that computation stops when $c$ is assigned a real number for the first time.

The Boolean expression $|F\big((a+b)/2\big)|\ge\eps$ is not quite a green light for the algorithm.
When it is violated for the first time, $c$ is still equal to \nil. But the equality $c=\nil$ is a green light.
With an eye on using Remark~\ref{rem:elegant}, we modify the draft program to the following program, our ``official'' program of an ASM $A$ representing the bisection algorithm.

\begin{eatab}
\> if\ \ $c=\nil$\>\>\>\> then\\[2pt]
\> \> if\>\>\ $F\big((a+b)/2\big) < -\eps$\>\>\>\>\>\>
   then $a := (a+b)/2$\\[2pt]
\> \> elseif\>\>\ $F\big((a+b)/2\big) >\eps$\>\>\>\>\>\>
   then $b := (a+b)/2$
   \\[2pt]
\>\> else\>\>\ $c := (a+b)/2$
\end{eatab}

\VA\ consists of the obligatory symbols, the symbols in \PA, and the unary relation symbol \Real.
In every initial state of $A$, \Real\ is (a copy of) the set of real numbers, the static functions of \PA\ have their standard meaning, and $c=\nil$.

Notice that Lemma~\ref{lem:simple} applies to \PA\ with \sn\ being $c := (a+b)/2$.
Taking this into account, the reversification procedure of \S\ref{sec:rev} gives us a reversible expansion $B$ of $A$ with the following program.

\begin{eatab}
\> if $c=\nil$ then\\
\>\> $\k:=\k+1 \prl$\\
\>\> if $F\big((a+b)/2\big) < -\eps$ then\\[2pt]
\>\>\> $a := (a+b)/2 \prl \F^1(\k+1):=\true
  \prl a_0(\k+1) := a$\\[2pt]
\>\> elseif $F\big((a+b)/2\big) >\eps$ then\\[2pt]
\>\>\> $b := (a+b)/2 \prl \F^2(\k+1):=\true
  \prl b_0(\k+1) := b$\\[2pt]
\>\> else $c := (a+b)/2$
\end{eatab}

This program can be simplified (and remain reversible).
Notice that
\begin{itemize}
\item if $\mathtt{Fire}^1(k+1) = \true$, then $\F^2(k+1) = \false$, the previous value of $b$ is the current value of $b$, and the previous value of $a$ is $2a-b$ where $a,b$ are the current values; and
\item if $\mathtt{Fire}^1(k+1) = \false$, then $\F^2(k+1) = \true$, the previous value of $a$ is the current value of $a$, and the previous value of $b$ is $2b-a$ where $a,b$ are the current values.
\end{itemize}
Thus, there is no need for functions $a_0, b_0$, recording the previous values of variables $a,b$, and there is no need for $\F^2$. We get:

\begin{eatab}
\> if $c=\nil$ then\\
\>\> $\k:=\k+1 \prl$\\
\>\> if $F\big((a+b)/2\big) < -\eps$ then\\[2pt]
\>\>\> $a := (a+b)/2 \prl \F^1(\k+1):=\true$\\[2pt]
\>\> elseif $F\big((a+b)/2\big) >\eps$ then $b := (a+b)/2$\\[2pt]
\>\> else $c := (a+b)/2$
\end{eatab}

The corresponding inverse algorithm may have this program:

\begin{eatab}
\ if $\k>0$ then\\[2pt]
\> $\k:=\k-1 \prl$ if $c\ne\nil$ then $c:=\nil$ \\[2pt]
\> $\parallel$\ if $\F^1(\k)=\true$ then
      $\big(\F^1(\k):=\false \parallel a:=2a-b\big)$ \\[2pt]
\>\phantom{$\parallel$} else $b:=2b-a$
\end{eatab}

\subsection{Linear-time sorting}
\label{sub:sort}\mbox{}

\smallskip
The information needed to reverse each step of the bisection algorithm is rather obvious; you don't have to use our reversification procedure for that.
Such information is slightly less obvious in the case of the following sorting algorithm.

For any natural number $n$, the algorithm sorts an arbitrary array $f$ of distinct natural numbers $<n$ in time $\le2n$.
Let $m$ be the length of an input array $f$, so that $m\le n$.
The algorithm uses an auxiliary array $g$ of length $n$ which is initially composed of zeroes.

Here is a simple illustration of the sorting procedure where $n=7$ and $f = \langle 3,6,0\rangle$.
Traverse array $f$ setting entries $g[f[i]]$ of $g$ to 1 for each index $i$ of $f$, i.e., setting $g[3], g[6]$ and $g[0]$ to 1, so that $g$ becomes $\langle 1,0,0,1,0,0,1\rangle$.
Each index $j$ of $g$ with $g[j]=1$ is an entry of the input array $f$.
Next, traverse array $g$ putting the indices $j$ with $g[j]=1$ --- in the order that they occur --- back into array $f$, so that $f$ becomes $\langle 0, 3, 6\rangle$. Voila, $f$ has been sorted in $m + n$ steps.

We describe an ASM $A$ representing the sorting algorithm. Arrays will be viewed as functions on finite initial segments of natural numbers. The nonobligatory function symbols in \VA\ are as follows.

\begin{enumerate}\setcounter{enumi}{-1}
\item Constants $m, n$ and variables $k,l$.
\item Unary dynamic symbols $f$ and $g$.
\item Binary static symbols\quad $<, +, -$\quad where $<$ is relational.
\end{enumerate}
In every initial state of $A$,
\begin{enumerate}\setcounter{enumi}{-1}
\item $m$ and $n$ are natural numbers such that $m\le n$, and $k=l=0$,
\item $f$, $g$ are arrays of lengths $m$, $n$ respectively, the entries of $f$ are distinct natural numbers $<n$, and all entries  of $g$ are zero,
\item the arithmetical operations $+,-$ and relation $<$ work as expected on natural numbers.
\end{enumerate}
In the following program of $A$, $k$ is the step counter, and $l$ indicates the current position in array $f$ to be filled in.

\begin{eatab}
\> if $k < m+n$ then\\[2pt]
\>\> $k := k+1 \prl$ \\[2pt]
\>\> if $k<m$ then $g(f(k)) := 1$ \\[2pt]
\>\> elseif $g(k-m)=1$ then $\big(f(l):=k-m \parallel l:=l+1\big)$
\end{eatab}

The reversification procedure of \S\ref{sec:rev} plus some simplifications described below give us a faithful reversible expansion $B$ of $A$ with a program

\begin{eatab}
\> if $k < m+n$ then\\[2pt]
\>\> $k := k+1 \prl$ \\[2pt]
\>\> if $k< m$
    then $\big(g(f(k)) := 1 \parallel g_1(k+1):=f(k)\big)$ \\[2pt]
\>\> elseif $g(k-m)=1$ then\\
\>\>\> $f(l):=k-m \prl l:=l+1 \prl f_0(k+1):=f(l)$
\end{eatab}

\noindent
We made some simplifications of \PB\ by discarding obviously unnecessary ancillary functions.
\begin{itemize}
\item It is unnecessary to record the firings of assignment $\sigma_2 = \big(g(f(k)):=1\big)$ because, in the states of $B$, the condition $\F^2(k)=\true$ is expressed by the inequality $k\le m$.
\item The ancillary function $g_0$ recording the previous values of $g$ is unnecessary because those values are all zeroes.
\item The final two assignments in \PA\ fire simultaneously, so that one fire-recording function, say $\F^3$, suffices. But even that one ancillary function is unnecessary because, in the states of $B$, the condition $\F^3(k)=\true$ is expressed by $m<k \land g(k-m-1)=1$.
\item The ancillary functions $f_1$ and $l_0$ recording the previous value of $l$ are unnecessary because we know that value, it is $l-1$.
\end{itemize}

The desired inverse algorithm $C$ may be given by this program:

\begin{eatab}
\> if $k>0$ then  \\[2pt]
\>\> $k:=k-1$ \\
\>\> $\parallel$ if $k\le m$ then
      $\big(g(g_1(k)):=0 \parallel g_1(k):=\nil\big)$ \\
\>\> $\parallel$ if $m<k$ and $g(k-m-1)=1$ then \\
\>\>\> $f(l):=f_0(k) \prl l:=l-1 \prl f_0(k):=\nil$
\end{eatab}

Obviously, $A$ is not reversible as is; its final state doesn't have information for reconstructing the initial $f$.
But do we need both remaining  ancillary functions?
Since $f_0$ is obliterated after the first $n$ steps of $C$, $f_0$ seems unlikely on its own to ensure reversibility.
But it does. The reason is that, after the first $n$ steps of $C$, the original array $f$ is restored.
Recall that $g_1(k)$ records the value $f(k-1)$ of the original $f$ for each positive $k\le m$, but we can discard $g_1$ and modify \PC\ to

\begin{eatab}
\> if $k>0$ then $k:=k-1$ \\
\>\> $\parallel$ if $k\le m$ then $g(f(k-1)):=0$ \\
\>\> $\parallel$ if $m<k$ and  $g(k-m-1)=1$ then \\
\>\>\> $f(l):=f_0(k) \parallel l:=l-1 \parallel f_0(k):=\nil$
\end{eatab}

Alternatively, we can discard $f_0$ but keep $g_1$. Indeed, the purpose of the assignment $f(l):=f_0(k)$ in \PC\ is to restore $f(l)$ to its original value. But recall that every $f(l)$ is recorded as $g_1(l+1)$. So we can modify \PC\ to

\begin{eatab}
\> if $k>0$ then $k:=k-1$ \\
\>\> $\parallel$ if $k\le m$ then
      $\big(g(g_1(k)):=0 \parallel g_1(k):=\nil\big)$ \\
\>\> $\parallel$ if $m<k$ and $g(k-m-1)=1$ then\\
\>\>\> $f(l):=g_1(l+1) \prl l:=l-1$
\end{eatab}

\subsection{External functions and Karger's algorithm}
\label{sub:karger}\mbox{}

Until now, for simplicity, we restricted attention to algorithms that are isolated in the sense that their computations are not influenced by the environment.
Actually, the analysis of sequential algorithms generalizes naturally and easily to the case when the environment can influence the computation of an algorithm \cite[\S8]{G141}. To this end, so-called external functions are used.

Syntactically, the item~(V3) in \S\ref{sub:syn} should be refined to say that a function symbol $f$ may be dynamic, or static, or external.
Semantically, external functions are treated as oracles%
\footnote{In that sense, our generalization is similar to the oracle generalization of Turing machines.}
When an algorithm evaluates an external function $f$ at some input \bx, it is the environment (and typically the operating system) that supplies the value $f(\bx)$.
The value is well defined at any given step of the algorithm; if $f$ is called several times, during the same step, on the same input \bx, the same value is given each time. But, at a different step, a different value $f(\bx)$ may be given.

To illustrate reversification involving an external function, we turn attention to Karger's algorithm \cite{Karger}.
In graph theory, a minimum cut of a graph is a cut (splitting the vertices into two disjoint subsets) that minimizes the number of edges crossing the cut.
Using randomization, Karger's algorithm constructs a cut which is a minimum cut with a certain probability. That probability is small but only polynomially (in the number of vertices) small.
Here we are not interested in the minimum cut problem, only in the algorithm itself.

\begin{terminology}
Let $G = (V,E)$ be a graph and consider a partition $P$ of the vertex set $V$ into disjoint subsets which we call \emph{cells}; formally $P$ is the set of the cells.
The $P$-\emph{ends} of an edge $\{x,y\}$ are the cells containing the vertices $x$ and $y$.
An edge is \emph{inter-cell} (relative to $P$) if its $P$-ends are distinct. \qef
\end{terminology}

Now we describe a version of Karger's algorithm that we call KA.
Given a finite connected graph $(V,E)$, KA works with partitions of the vertex set $V$, one partition at a time, and KA keeps track of the set Inter of the inter-cell edges.
KA starts with the finest partition $P = \big\{ \{v\}: v\in V\big\}$ and $\Inter = E$.
If the current partition $P$ has $>2$ cells, then \Inter\ is nonempty because the graph $(V,E)$ is connected. In this case, KA selects a random inter-cell edge $e$, merges the $P$-ends $p,q$ of $e$ into one cell, and removes from \Inter\ the edges in $\big\{\set{x,y}: x\in p \land y\in q\big\}$.
The result is a coarser partition and smaller \Inter.
When the current partition has at most two cells, the algorithm stops.

Next we describe an ASM $A$ representing KA.
There are many ways to represent KA as an ASM.
Thinking of the convenience of description rather than implementation of KA, we chose to be close to naive set theory. 
Let $U$ be a set that includes $V$ and all subsets of $V$ and all sets of subsets of $V$ (which is much more than needed but never mind).
The relation $\in$ on $U$ has its standard meaning; the vertices are treated as atoms (or urelements), not sets.

The nonobligatory function symbols of \VA\ are as follows.

\begin{enumerate}\setcounter{enumi}{-1}
\item Nullary variables $P$ and \Inter.
\item Unary static symbols $V, E, |.|$, and a unary external symbol $R$.
\item Binary static symbols $>, -$, Merge, and \Intra, where $>$ is relational.
\end{enumerate}

In every initial state of $A$,
\begin{itemize}
\item $V$, $U$ and $\in$ are as described above (up to isomorphism). $|s|$ is the cardinality of a set $s$, and $-$ is the set-theoretic difference. The relation $>$ is the standard ordering of natural numbers
\item $E$ is a set of unordered pairs \set{x,y} with $x,y\in V$ such that the graph $(V,E)$ is connected. $P$ is the finest partition $\big\{ \set{v}: v\in V\big\}$ of $V$. $\Inter = E$.
\item If $e \in E$, $S$ is a partition of $V$, and $p,q$ are the $S$-ends of $e$, then
  \begin{itemize}
  \item $\Merge(e,S) = (S-\set{p,q}) \cup \set{p\cup q}$, and
  \item $\Intra(e,S) = \big\{ \set{x,y}: x\in p \land y\in q\big\}$.
  \end{itemize}
\end{itemize}
\noindent
The external function $R$ takes a nonempty set and returns a member of it.
The program of $A$ can be this:

\begin{eatab}
\> if $|P|>2$ then \\[2pt]
\>\> $P:= \Merge(R(\Inter),P)$\\
\>\> $\parallel\ \Inter := \Inter - \Intra(R(\Inter),P)$
\end{eatab}

Now we apply the reversification procedure of Theorem~\ref{thm:rev}, taking Remark~\ref{rem:elegant} into account.
We also take into account that both assignments fire at every step of the algorithm and so there is no need to record the firings.
This gives us a faithful reversible expansion $B$ of $A$ with a program

\begin{eatab}\small
\> if $|P|>2$ then \\
\>\> $\k:=\k+1 \prl$ \\
\>\> $P:= \Merge(R(\Inter),P) \prl P_0(\k+1) := P \prl$ \\
\>\> $\Inter := \Inter - \Intra(R(\Inter),P) \prl \Inter_0(\k+1):=\Inter$
\end{eatab}

\noindent
The corresponding inverse ASM $C$ may be given by the program

\begin{eatab}
\> if $\k>0$ then \\
\>\> $\k:=\k-1 \prl$\\
\>\> $P:=P_0(\k) \prl P_0(\k):=\nil \prl$ \\
\>\> $\Inter:=\Inter_0(\k) \prl \Inter_0(\k):=\nil$
\end{eatab}

\begin{remark}
A custom crafted faithful expansion may be more efficient in various ways.
For example, instead of recording the whole $P$, it may record just one of the two $P$-ends of $R(\Inter)$. This would require a richer vocabulary.
\end{remark}

\section{Conclusion}

We have shown how to reversify an arbitrary sequential algorithm $A$ by gently instrumenting $A$ with bookkeeping machinery.
The result is a step-for-step reversible algorithm $B$ whose behavior, as far as the vocabulary of $A$ is concerned, is identical to that of $A$.

We work with an ASM (abstract state machine) representation of the given algorithm which is behaviorally identical to it.
The theory of such representation is developed in \cite{G141}, and the practicality of it has been amply demonstrated.

\end{document}